\documentclass[11pt]{article}

\usepackage{palatino}
\usepackage{mathpazo}
\usepackage{graphicx}
\usepackage{color}

\usepackage[margin=1in]{geometry}
\usepackage{amssymb,amsthm,amsmath}

\DeclareMathOperator{\conv}{conv}

\newtheorem{theorem}{Theorem}
\newtheorem{cor}[theorem]{Corollary}

\newtheorem{lemma}[theorem]{Lemma}

\newcommand{\ket}[1]{|#1\rangle}

\newcommand{\op}[2]{|#1\rangle \langle #2|}

\newcommand{\tr}{{\rm tr}}

\renewcommand{\>}{\rangle}
\newcommand{\<}{\langle}

\def\X{\mathcal{X}}
\def\Y{\mathcal{Y}}

\def\L{\mathcal{L}}
\def\P{\mathcal{P}}

\newlength{\blank}
\newenvironment{proof-of}[1][{\hspace{-\blank}}]{{\medskip\noindent\textit{Proof~{#1}.\ }}}{\hfill\qedsymbol}

\begin{document}

%---------Metadata---------%
\title{\bf LOCC protocols with bounded width per round \protect\\ optimize convex functions}
\author{
  Debbie Leung\footnote{%
    Institute for Quantum Computing and
    Department of Combinatorics and Optimization,
    University of Waterloo,
    Waterloo, Ontario, Canada.}
  \and
  Andreas Winter\footnote{ICREA \& F\'{i}sica Te\'{o}rica: Informaci\'{o} i Fen\'{o}mens Qu\`{a}ntics, Departament de F\'{i}sica, Universitat Aut\`{o}noma de Barcelona, ES-08193 Bellaterra (Barcelona), Spain}
  \and
  Nengkun Yu$^{*,}$\footnote{Centre for Quantum Software and Information,
   Faculty of Engineering and Information Technology, University of
   Technology Sydney NSW 2007, Australia.}
}

\date{20 April 2019}

\maketitle

%---------Abstract---------%
\begin{abstract}
We start with the task of discriminating finitely many multipartite
quantum states using LOCC protocols, with the goal to optimize
the probability of correctly identifying the state.
We provide two different methods to show that finitely many
measurement outcomes in every step are sufficient for
approaching the optimal probability of discrimination.
In the first method, each measurement of an optimal LOCC protocol,
applied to a $d_{\rm loc}$-dim local system, is replaced by one with
at most $2d_{\rm loc}^2$ outcomes, without changing the probability of
success.
In the second method, we decompose \emph{any} LOCC protocol into a
convex combination of a number of ``slim protocols'' in which each
measurement applied to a $d_{\rm loc}$-dim local system has at most
$d_{\rm loc}^2$ outcomes.
To maximize any convex functions in LOCC (including 
the probability of state discrimination or fidelity of state
transformation), an optimal protocol can be replaced by the best slim
protocol in the convex decomposition without using shared randomness.
For either method, the bound on the number of outcomes per measurement
is independent of the global dimension, the number of parties, the
depth of the protocol, how deep the measurement is located, and
applies to LOCC protocols with infinite rounds,
and the ``measurement compression'' can be done ``top-down'' -- 
independent of later operations in the LOCC protocol.
The second method can be generalized to implement LOCC instruments
with finitely many outcomes: if the instrument has $n$ coarse-grained 
final measurement outcomes,
global input dimension $D_0$ and global output
dimension $D_i$ for $i=1,\cdots,n$ conditioned on the $i$-th outcome,
then one can obtain the instrument 
as a convex combination of no more than $R=\sum_{i=1}^n D_0^2
D_i^2 - D_0^2 + 1$ slim protocols; in other words, $\log_2 R$ bits of shared
randomess suffice. 
\end{abstract}

\clearpage

%%%%%%%%%%%%%%%%%%%%%%%%%%%%%%%%%%%%%%%%%%%%%%%%%%%%%%%%%%%%%%%%%%%%%%
\section{Introduction}
\label{intro}

For a multi-partite quantum system, the class of operations that can
be implemented by composing local operations on each individual part
and classical communication between the parts is shorthanded LOCC.
This class originates from the seminal work by Peres and Wootters
\cite{PW91} and its
importance has been manifest in many subsequent results, such as
\cite{teleportation,BDSW96}.
One motivation for the LOCC class is the operational difficulty of
long range quantum communication.  From a more fundamental
perspective, LOCC is precisely the class of operations that can be
implemented \emph{without} shared entanglement; therefore LOCC provides a
natural framework to study quantum nonlocality and entanglement.
Understanding the power and the limitation of LOCC operations
is one of the main goals of quantum information theory.
In particular, we say that an information processing task exhibits
nonlocality when it can be accomplished using global operations but
not by LOCC operations.
While the class of LOCC operations is well motivated, it does not
have a succinct mathematical characterization, and the complexity
grows rapidly with the number of rounds of communication.

In this paper, we first consider the quantum state discrimination problem,
in which a list of quantum states is fixed in advance.  A referee
chooses a state from the list, prepares a copy, and distributes it to
the discriminating party (or parties), whose goal is to identify which
state has been prepared by the referee.
In some situations the prepared state can be identified without error.
Otherwise, one can relax the problem by assuming that the referee
picks a state from the list according to some pre-determined
distribution known to the parties, and their goal is to maximize the
\emph{success probability} (i.e., the probability of correctly
identifying the state).

The quantum state discrimination problem provides a fruitful line of
studies in our understanding of LOCC and nonlocality.  If restricting
the players to LOCC operations strictly decreases their probability of
success, the problem exhibits nonlocality.
A partial list of references on this problem can be found in
% \cite{MP95,BDFM+99,BDMS+99,WSHV00,QDH01,QDH02,EW02,GKRS+01,HSSH03,GKRS04,FAN04,NAT05,WAT05,HMM+06,COH07,DFJY07,DFXY09,MW08,MWW09,Ehi11,BAN11,YDY11b,BGK11,YDY11a,ALE12,YDY14,BCJ+14,Coh15}.
%
\cite{MP95}-\cite{Coh15}. 
Because there is no succinct description of all possible LOCC
discrimination strategies, one widely used approach is to use a larger 
set of operations to study the limitation on the distinguishability
power of LOCC, for instance, separable operations (SEP) or PPT-preserving
operations
\cite{QDH01,QDH02,EW02,MW08,MWW09,YDY11b,YDY14,BCJ+14}.

More recent studies have found useful structural and topological
properties of LOCC 
% \cite{Chi11,CLM+14,Coh15,Coh17}.  
\cite{Coh15}-\cite{Coh17}.  In \cite{Chi11,CLL12,CLM+14}, the set of LOCC
operations is shown to be not closed. Explicit entanglement
transformation tasks are given in \cite{Chi11,CLL12,CLM+14} that are
\emph{provably} not accomplished by any finite round LOCC protocol but
that can be approximated with arbitrary precision when the number of
communication rounds increases. Consequently, we cannot assume that
an LOCC protocol has a finite number of communication rounds, called
the \emph{depth} of the protocol.  Even after restricting to finite depth
protocols for the task, it is not clear apriori whether more and more
outcomes in some intermediate measurement in the protocol can
approximate the ideal task better and better.  
The total width of the protocol refers to the total number of
measurement outcomes in the protocol.  A related concept is the width
per measurement, which is the number of outcomes for each measurement.
An LOCC protocol can be represented by a decision
tree where vertices represent operations, and edges represent
measurement outcomes, thus the names depth and total width of a protocol.
(Note that the width per measurement is the degree of the root vertex,
or the degree minus $1$ for other vertices.)
Reference \cite{CLM+14} shows that the set of all LOCC protocols with
a constant number of rounds of communication is \emph{compact}, and
provides an upper bound on the number of measurement outcomes.

In this paper, we focus on the width of LOCC protocols.  
In the more specialized context of multipartite quantum state
discrimination, we show that LOCC protocols with finite width per
measurement are sufficient to achieve the optimal probability of
success in quantum state discrimination under LOCC with two
different methods.
The second method extends to optimizing any convex function,
including the probability of state discrimination, and the fidelity of
state transformation.
Both results apply to LOCC (as defined in Section 2.2 of
\cite{CLM+14}).  Informally, this class includes infinite round LOCC
protocols that can be approximated better and better by adding more
and more rounds of communications (without changing the earlier
steps).  Each protocol in this class can be represented by a tree 
that can be infinite.
Each measurement is replaced by one with few outcomes in a
``top-down'' manner -- starting from the root (where the protocol
begins), we replace each measurement as we move down the tree (as the
protocol progresses) in a way \emph{independent} of how deep the protocol
will be executed.  The original task can be approximated better and
better by going deeper in the resulting single finite-width
infinite-depth protocol.

Our first method converts every measurement in the protocol (possibly
with infinitely many outcomes) into one with no more than 
$2d_{\rm loc}^2$ outcomes where $d_{\rm loc}$ is the dimension of the local
system being measured.  If the protocol is finite with $\ell$ rounds
of communication and $d$ is the largest of the local dimensions, the
total width of the protocol is upper bounded by $2^\ell d^{2\ell}$.
Our second method converts every measurement into one with no more
than $d_{\rm loc}^2$ outcomes.  If the protocol is finite with $\ell$
rounds of communication and $d$ is the largest of the local
dimensions, the total width of the protocol is upper bounded by
$d^{2\ell}$.

Both methods are constructive, and rely on Caratheodory's theorem.
They are simpler than the compression given by \cite{CLM+14} for
finite LOCC protocols, and our bounds are tighter (independent of the
global dimension and the number of parties, independent of how deep
the protocol has run, and has lower degree in the dimension). Most 
importantly, our compression is top-down.

The second method also implements any LOCC instrument with the
aforementioned width per measurement by using additional shared
randomness.  If the protocol has $n$ coarse-grained final measurement
outcomes, global input dimension $D_0$ and global output dimension
$D_i$ for $i=1,\cdots,n$ conditioned on the $i$-th outcome, then
$\log_2 R$ bits of shared randomess suffice where $R=\sum_{i=1}^n
D_0^2 D_i^2 - D_0^2 + 1$.

Towards the final stages of preparing this manuscript, we learnt of
related results by Cohen \cite{Coh19}, who shows that any LOCC quantum
operation ${\cal E}$ with potentially unbounded width can be converted
to one with finite width per round.
The number of measurement outcomes per round is upper bounded by
$\min(\kappa^2,\kappa^2 + d_{\rm loc}^2 - \chi)$ where $\kappa$ is the
global Kraus rank of ${\cal E}$ and $\chi = \dim \left( {\rm span}
\{K_i^\dagger K_j\}_{i,j=1}^\kappa \right)$ where $K_i$'s are the Kraus
operators of ${\cal E}$.
Cohen's method preserves the quantum operations.
In comparison, our methods for optimizing concave functions need not
preserve the quantum operations, but may have a tighter bound on the
width per measurement in some regime.  (For example, to discriminate
many states shared by a large number of parties, each holding a small
dimensional system, the Kraus rank scales as the number of states
which is much larger than the local dimension.)
Our second method can be extended to preserve the quantum operations
by using finite amount of shared randomness.
Our correspondence with Cohen had inspired improvements in aspects of 
our second method (including a discussion on the shared randomness,
and the application of state transformation).  
It remains unclear how closely the approaches in these two papers are
related, and whether the techniques can be combined to obtain better
results.

In Section \ref{primer} we cover the mathematical background and
define notations and concepts required for the discussion.
The main results are presented in Section \ref{sec:result}.

%%%%%%%%%%%%%%%%%%%%%%%%%%%%%%%%%%%%%%%%%%%%%%%%%%%%%%%%%%%%%%%%%%%%%%
\section{Preliminaries}
\label{primer}

% In this section, we introduce relevant background and notation in quantum
% information for our discussion.

The term {\it Hilbert space} here refers to any \emph{finite}
dimensional semidefinite inner product space over the complex numbers.
Let $\X$ be an arbitrary Hilbert space. A \emph{pure} quantum state of
$\X$ is a \emph{normalized} vector $\ket{\Psi}\in \X$.  A quantum
mechanical system is associated with a {\it Hilbert space}
and we refer to both the system and the space with the same notation.
A composite system is associated with the tensor product of the
{\it Hilbert spaces} associated with the parts.

The space of linear operators mapping $\X$ to $\Y$ is denoted by
$\mathcal{L}(\X,\Y)$, while $\mathcal{L}(\X)$ is the shorthand for
$\mathcal{L}(\X,\X)$. We use $I_{\X}$ to denote the identity operator
on $\X$, and often omit the system label $\X$.  The adjoint (or
Hermitian transpose) of $A\in\mathcal{L}(\X,\X)$ is denoted by
$A^{\dag}$.  The notation $A\geq 0$ means that $A$ is positive
semidefinite, and more generally $A\geq B$ means that $A - B$ is
positive semidefinite.  The positive square root of $A^{\dag}A$ is
denoted by $|A|=\sqrt{A^\dag A}$.

A (general) quantum state is specified by its density operator
$\rho\in \mathcal{L}(\X)$, which is a positive semi-definite operator
with trace one.  The density operator of a pure state $\ket{\psi}$ is
simply the projector $\psi:=\op{\psi}{\psi}$.

A quantum measurement $\mathcal{M}$ with input system $\mathcal{\X}$
and output system $\mathcal{\Y}$ is specified by a POVM $(A_1^\dagger
A_1, A_2^\dagger A_2, \cdots)$ where each $A_i \in
\L(\mathcal{\X},\mathcal{\Y})$ and $\sum_i A_i^\dagger A_i = I$.
If the initial state being measured is $\rho$,
$$\mathcal{M}(\rho) = \sum_i A_i \rho A_i^\dagger \otimes \op{i}{i}$$
where $i$ is the measurement outcome, and $A_i \rho A_i^\dagger$ is
the corresponding unnormalized postmeasurement quantum state whose
norm $\tr A_i \rho A_i^\dagger = A_i^\dagger A_i \rho$ gives the
probability of obtaining outcome $i$.
More generally, each $A_i$ can take the input system $\mathcal{\X}$
to an output system $\mathcal{\Y}_i$, where the $\mathcal{\Y}_i$'s
may not have the same dimension.  

The most general quantum operation ${\cal E}$ with input system
$\mathcal{\X}$ and output system $\mathcal{\Y}$ acts as ${\cal
  E}(\rho) = \sum_i A_i \rho A_i^\dagger$, where $\sum_i A_i^\dagger
A_i = I$.  An instrument with input system $\mathcal{\X}$ acts as
${\cal I}(\rho) = \sum_i {\cal I}_i (\rho) \otimes |i\>\<i|$ where
each ${\cal I}_i$ is a completely positive map, and $\sum_i {\cal
  I}_i$ is trace preserving.  A measurement is a fine-grained
instrument in which all CP maps has Kraus rank 1.  

Consider an ensemble of quantum states
$$ S=\{p_1\rho_1,\cdots,p_n\rho_n\}\subset \L(\mathcal{\X}) $$
where $\rho_k$ are normalized states and $p_k\geq 0$, $\sum_k p_k\leq 1$.
Then $\sum_k p_k$ is called the probability of the ensemble $S$.
If $\sum_k p_k=1$, the ensemble is called \emph{normalized}.

Throughout this paper, we focus on multipartite quantum systems of the form,
$$\mathcal{\X}=\mathcal{\X}_1\otimes\mathcal{\X}_2\otimes\cdots\otimes\mathcal{\X}_m.$$
LOCC, or local operations and classical communication, on this system
$\mathcal{X}$, is the set of operations such that each party is
restricted to performing quantum operations on their individual local
systems and they may communicate classical information (w.l.o.g.,
measurement outcomes) to the other parties.  Any LOCC operation can be
decomposed into rounds; at each round, one party applies a quantum
operation on his/her local system and broadcasts a classical message
to all other parties.  An LOCC protocol can have infinitely many
rounds of communication.  See \cite{CLM+14} for detail.

For any LOCC protocol $\P$ on $\mathcal{X}$ (potentially infinitely
wide and with infinitely many rounds, but one that can be approximated
by adding rounds and can be represented by a tree), denote its
$\ell$-round prefix by $\P_\ell$.  For any ensemble
$S=\{p_1\rho_1,\cdots,p_n\rho_n\}\subset \mathcal{\X}$, the
probability of successful discrimination by $\P$ can be defined as
follows,
$$P(\P,S)=\lim_{\ell \rightarrow \infty} P(\P_\ell,S),$$ where $P(\P_\ell,S)$
denotes the probability of successful discrimination of $S$ by $\P_\ell$.
As $\P_\ell$ can potentially have infinite width, $P(\P_\ell,S)$ is
similarly defined as a limit.

In our second method, we use the following notation and
terminology derived from \cite{CLM+14}.  Any protocol in LOCC can
be represented as a possibly infinite tree.  The protocol starts at
the root and moves through the tree along edges, always further away
from the root.  Each vertex $v$ is associated with an instrument 
applied to a local system ${\cal H}_{\rm loc}$ held by one party.  
We can write this instrument as  
$${\cal L}_v(\rho) = \sum_{w:{\rm child~of}~v} 
{\cal L}_{(w,v)}(\rho) \otimes |(w,v)\> \< (w,v)|,$$ 
where each outgoing edge $(w,v)$ 
is associated with a CP map ${\cal L}_{(w,v)}$ acting on
${\cal H}_{\rm loc}$, such that $\sum_{w:{\rm child~of}~v} {\cal L}_{(w,v)}$ 
is trace-preserving.  
Each vertex $v$ at depth $\ell$ can be reached by a unique path
from the root $r$, $(r,v_1), (v_1,v_2), \cdots, (v_{\ell-1},v)$
and the vertex is associated with a ``cumulative'' CP map 
$${\cal N}_v := {\cal L}_{(v,v_{\ell-1})} \circ \cdots
\circ {\cal L}_{(v_2,v_1)} \circ {\cal L}_{(v_1,r)}.$$

The LOCC protocol implements an LOCC instrument ${\cal L}$ which can
be specified as follows.  Consider any function $f{:}L \rightarrow O$
from the set of leaves $L$ of the tree, to a set of outcomes $O$.  The
instrument ${\cal L}$ is implemented by running the protocol from the
root until arriving at a leaf $v$ and then outputting $f(v)$.  For
each $o \in O$, let ${\cal I}_o = \sum_{v\in f^{-1}(o)} {\cal N}_v$.
The instrument implemented by the LOCC protocol is given by ${\cal
  I}(\rho) = \sum_o {\cal I}_o(\rho) \otimes |o\>\<o|$.  For a finite
tree, $\sum_o {\cal I}_o$ is trace preserving.  For an infinite tree
we have to make an extra assumption that for every input the protocol
terminates with probability 1.

We can fine-grain an LOCC protocol by breaking up CP maps
associated with edges into Kraus-rank-1 CP maps, enlarging the tree,
and modifying the coarse-graining function $f$ accordingly, without
changing the instrument implemented by the protocol. We call such a
protocol \emph{fine-grained}.

%%%%%%%%%%%%%%%%%%%%%%%%%%%%%%%%%%%%%%%%%%%%%%%%%%%%%%%%%%%%%%%%%%%%%%
\section{Main Result}
\label{sec:result}
\subsection{The first method and resulting bounds}
% Our main result is given by the following theorem.

\begin{theorem}
\label{thm:main}
Suppose an ensemble of multipartite quantum states
$S=\{p_1\rho_1,\cdots,p_n\rho_n\}\subset\L(\otimes_{j=1}^m\mathcal{\X}_j)$
with $p_k \geq 0$, $\sum_k p_k=1$ can be distinguished by some LOCC
protocol $\mathcal{P}$ with success probability $t$.  Then there
exists an LOCC protocol $\mathcal{P}'$ achieving the same success
probability $t$ but in which each measurement requires at most $2
d_{\rm loc}^2 $ outcomes, where $d_{\rm loc}$ is the dimension of
the local system measured.  If $\mathcal{P}$ has finitely many rounds
of communication $\ell$, the total width can be bounded by $2^\ell
d^{2\ell}$, where $d$ is the maximum local dimension.
\end{theorem}

We first discuss informally the intuition behind the 
constructive proof.  We obtain the bound by recursively ``compress''
an arbitrary measurement in the protocol while preserving the success
probability, depth of the protocol, and the induced post-measurement
ensembles.
The compression for a measurement is done in several steps:
\begin{enumerate}
\item Show that the measurement can be performed in two stages (as
a composition of two measurements).
\item Show that the first stage measurement can be modified
to ``equalize'' the success probability on the induced
post-measurement ensemble for each measurement outcome.  This step
preserves the success probability of the protocol.
\item A convexity argument shows that all but a finite number of measurement
outcomes can be dropped for the first stage while preserving the
probability of success.
\item The modifications in steps 2-3 are compatible with the two stage
implementation of the original measurement.  So, the second stage
measurement is applied for each of the finitely many outcomes in stage 1.  
This preserves the depth and success probability of the protocol.
Furthermore, all the subsequent steps in the original protocol are
unaffected.
\end{enumerate}
The following lemma will be needed for steps 1, 2, and 4 above.
%-------------------------------------------------------------------
\begin{lemma}
\label{lem:matrix-sum}
For any pair of matrices $X,Y$ of the same width, there exists a matrix
$C$ of the same size as $X$ and a matrix $D$ of the same size as $Y$ such
that
\begin{eqnarray*}
C\sqrt{X^{\dag}X+Y^{\dag}Y}=X,\\
D\sqrt{X^{\dag}X+Y^{\dag}Y}=Y,\\
C^{\dag}C+D^{\dag}D=I.
\end{eqnarray*}
\end{lemma}

\begin{proof-of}[of~Lemma \ref{lem:matrix-sum}]
If $X^{\dag}X+Y^{\dag}Y$ is nonsingular, then the choices
\begin{eqnarray*}
C=X(X^{\dag}X+Y^{\dag}Y)^{-1/2}, \label{eq:c} \\
D=Y(X^{\dag}X+Y^{\dag}Y)^{-1/2}. \label{eq:d}
\end{eqnarray*}
imply $C^{\dag}C+D^{\dag}D=I$.  Otherwise, replace
$(X^{\dag}X+Y^{\dag}Y)^{-1/2}$ by its restriction on the support of
$X^{\dag}X+Y^{\dag}Y$ in the above expressions of $C$ and $D$,
and add to the expression of $C$ a projector onto the null space of
$X^{\dag}X+Y^{\dag}Y$.
\end{proof-of}

\begin{proof-of}[of~Theorem \ref{thm:main}]
Without loss of generality, the LOCC protocol $\mathcal{P}$
has the following form.
In the first round, one of the parties (w.l.o.g., the first party) applies
a measurement $\mathcal{M}$ with POVM $(A_1^\dagger A_1,A_2^\dagger
A_2,\cdots)$, possibly with infinitely many outcomes. Then, the party
broadcasts the measurement outcome.  In the second round, another
party applies another local measurement that can depend on the first
outcome, and broadcasts the second outcome.  This goes on, either for
some finitely many rounds, say, $\ell$, or indefinitely.

We construct $\mathcal{P}'$ from $\mathcal{P}$ as follows.

For simplicity, we focus on the compression method on the first
measurement $\mathcal{M}$.  Every possible measurement outcome $i$
induces a post-measurement ensemble
$S_i=\{p_1A_i\rho_1A_i^{\dag},\cdots,p_nA_i\rho_nA_i^{\dag}\}$.  
Denote the probability of the ensemble $S_i$ by $q_i$.  We focus on
the set of $i$'s for which $q_i > 0$.  Conditioned on the outcome $i$,
$S_i/q_i$ is a normalized ensemble, with some probability of
successful discrimination $t_i$ (see Section \ref{primer}).  The
$t_i$'s are related to the total success probability by
\begin{equation}
t = \sum_i q_i t_i \,.
\end{equation}

We first show that $\mathcal{M}$ can be performed in two stages.  and
that the first stage can be modified to some $\mathcal{M}'$ to
equalize the success probability for each outcome.
Assume without loss of generality,
\begin{equation}
  t_1\geq t_2\geq t_3 \geq \cdots
  % \geq t_r > t_{r+1} \geq \cdots
\end{equation}
If $t_i = t$ for all $i$, we are done.  So, suppose there exists some
$k$ such that $t > t_k$, which also implies $t_1 > t$.
There exists $0 < s$ such that
\begin{equation}\label{solution}
\frac{q_1t_1+sq_kt_k}{q_1+sq_k}=t \,.
\end{equation}
To see this, note that $t = (1-\lambda) t_1 + \lambda t_k$ for some
$\lambda \in (0,1)$.  Then, it suffices for $\frac{s q_k}{q_1 + sq_k}
= \lambda$, which holds if
\begin{equation}\label{s}
s = \frac{\lambda q_1}{(1-\lambda) q_k}\,.
\end{equation}
We now consider the two cases $s \leq 1$ and $s > 1$ separately.

If $s \leq 1$, let $B=\sqrt{A_1^{\dag}A_1+sA_k^{\dag}A_k}$.  
Consider the induced ensemble
$BS=\{p_1B\rho_1B^{\dag},\cdots,p_nB\rho_nB^{\dag}\}$.  
The probability of the ensemble $BS$ is equal to $q_1+sq_k$.  
We now show that $BS/(q_1+sq_k)$ has success probability
$(q_1t_1+sq_kt_k)/(q_1+sq_k)$, which equals to $t$.
To see this:
\vspace{-2ex}
\begin{quote}
Consider a modification to $\mathcal{M}$ by replacing
$A_1$ and $A_k$ by $B$ and $\sqrt{1-s} A_k$ respectively.  Call the
resulting measurement $\tilde{\mathcal{M}}$.
If $BS/(q_1+sq_k)$ has probability of success
greater than $t$, replacing $\mathcal{M}$ by $\tilde{\mathcal{M}}$ in
$\mathcal{P}$ outperforms $\mathcal{P}$, contradicting its optimality.
Conversely, consider the application to the ensemble $BS$ a binary
measurement $\mathcal{N}$ defined by the POVM $(C^\dagger C,D^\dagger D)$
where $C,D$ are obtained as in Lemma \ref{lem:matrix-sum} with $X=A_1$
and $Y=\sqrt{s} A_k$.  The lemma guarantees that $\mathcal{N}$ is a
valid measurement on the postmeasurement space of $\tilde{\mathcal{M}}$. 
Using Lemma \ref{lem:matrix-sum} to 
combine the effects due to $\tilde{\mathcal{M}}$ and $\mathcal{N}$, 
one can see that 
the outcome of $\mathcal{N}$ corresponding to $C^\dagger C$ induces the
postmeasurement ensemble $S_1$ while the outcome corresponding to
$D^\dagger D$ induces the postmeasurement ensemble $s S_k$.
Therefore, $BS/(q_1+sq_k)$ has success probability at least $(q_1 t_1
+ s q_k t_k)/(q_1+sq_k)$ which is equal to $t$ (see (\ref{solution})).
\end{quote}
\vspace{-2ex}
Observe that modifying $\mathcal{M}$ to $\tilde{\mathcal{M}}$ replaces
$t_1$ by $t$, $q_1$ by $q_1 + s q_k$, $q_k$ by $(1-s) q_k$, while
$t_k$ is left unchanged.  Also, applying $\mathcal{N}$ after
$\tilde{\mathcal{M}}$ gives the original $\mathcal{M}$.

If $s>1$, (\ref{s}) can be rewritten as $\frac{1}{s} =
\frac{(1-\lambda) q_k}{\lambda q_1}$.  A similar argument holds (and
we do not repeat it here), with
$A_1$ and $A_k$ interchanged.  In this case, we replace $A_k$ by
$B'=\sqrt{\frac{1}{s}A_1^{\dag}A_1+A_k^{\dag}A_k}$ and $A_1$ by
$\sqrt{1-\frac{1}{s}} A_1$ to obtain $\tilde{\mathcal{M}}$, and
$t_k$ is replaced by $t$.

In either case, modifying $\mathcal{M}$ into $\tilde{\mathcal{M}}$
\emph{strictly} increases the probability to have an induced
postmeasurement ensemble that has probability of success equal to $t$.
We repeat this modification until all postmeasurement ensembles have
probability of success equal to $t$ (a property we need later when we
reduce the number of outcomes).
The resulting measurement is the desired first stage measurement
$\mathcal{M}'$, say, with POVM $(B_1^\dagger B_1, B_2^\dagger B_2,
\cdots)$.
Also, from the above discussion, for each outcome of
$\mathcal{M}'$, there is a subsequent second stage binary measurement
that completes $\mathcal{M}$.

In the next step, we replace $\mathcal{M}'$ by $\mathcal{M}''$ which has
only $d_1^2$ measurement outcomes, where $d_1$ is the dimension of
the system measured (and held by the first party).  For this we use
Carath\'{e}odory's Theorem (which has a constructive proof):

\begin{lemma}[Carath\'{e}odory's Theorem~\cite{Rockafellar-1996a}]
\label{lem:ct}
Let $H$ be a subset of $\mathbb{R}^n$ and $\conv(H)$ its convex hull.
Then any $x \in \conv(H)$ can be expressed as a convex combination of
at most $n+1$ elements of $H$.
\end{lemma}

To rewrite the sum
$\sum_i B_i^{\dag} B_i = I$,
note that $\sum_i u_i \frac{B_i^{\dag} B_i}{\tr B_i^{\dag} B_i} =
\frac{I}{d_1}$, where $u_i = \frac{\tr B_i^{\dag}
  B_i}{d_1}$ form a distribution.  
So, 
we can apply Carath\'{e}odory's Theorem with 
$H = \{ \frac{B_i^{\dag} B_i}{\tr B_i^{\dag} B_i} \}_i$ which is a subset 
of all trace 1
$d_1 \times d_1$ hermitian matrices with $n=d_1^2-1$, and obtain $I$ as a sum
of at most $d_1^2$ operators, each is a positive multiple of some
$B_{i}^\dagger B_{i}$. 
This new sum defines a new first stage
measurement $\mathcal{M}''$, which is similar to of $\mathcal{M}'$,
but now only $d_1^2$ outcomes are possible. For each outcome of
$\mathcal{M}''$, the induced postmeasurement ensemble is the same as
in $\mathcal{M}'$ and has success probability $t$.

Finally, for each outcome of $\mathcal{M}''$, we apply the binary
measurement that brings the postmeasurement ensemble back to that of
$\mathcal{M}$.  The total number of outcomes is at most $2d_1^2$.
This completes the compression of the first measurement $\mathcal{M}$.

After the first round of communication, conditioned on each outcome,
the parties now hold a new, normalized, ensemble, and they try their
best to discriminate it (with $\ell-1$ rounds of communication if
$\mathcal{P}$ has $\ell$ rounds).  A similar compression can now be
applied to the next measurement.  Repeating the process, each
measurement in the protocol has no more than $2 d_{\rm loc}^2$
outcomes.  If $\mathcal{P}$ has $\ell$ rounds, the total number of
outcomes is at most $2^\ell d^{2\ell}$, where
$d=\max\{d_1,d_2,\cdots,d_m\}$ is the maximum local dimension.
\end{proof-of}

Note that without the constraint of being in an LOCC protocol, a
measurement on a $d$-dimensional system can be compressed to
$d^2$ outcomes.
This bound $2d^2$ shows that to optimize state discrimination in
LOCC, about twice as many outcomes (or one additional bit of 
communication) are sufficient.
This is independent on the number of rounds (and finite
or not), how deep the parties have executed the protocol, the number of
parties or the total dimension of the system, and not on the
number of states in $S$.
In comparison, in \cite{CLM+14} each measurement in round $\ell$ out of
$r$ has at most $n D^{4(r-\ell+1)}$ outcomes where $n$ is the number of
outcomes, after coarse-graining, at the end of the protocol (which 
is $|S|$ for state
discrimination) and $D$ is the global dimension.
Most importantly, this bound diverges when $r$ diverges.

If we apply Carath\'{e}odory's Theorem (Lemma \ref{lem:ct}) to the
original POVM $\{A_i^\dagger A_i\}$ to reduce the number of
measurement outcomes, the probability of discrimination need not be
preserved.  We introduce the first stage modification to equalize the
probability of correct discrimination for each outcome, and need to
add a second stage measurement, thereby getting an additional 
factor of $2$ in the bound $2d_1^2$.
The next method improves the bound to $d_1^2$.  It is based on the
limited number of outcomes for extremal measurement and the limited
size of the support of extremal distributions, both of which are
corollaries of Carath\'{e}odory's Theorem.

\subsection{The second method and improved bounds}
%%%%%%%%%%%%%%%%%%%%%%%%%%%%%%%%%%%%%%%%%%%%%%%%%%%%%%%%%%%%%%%%%%%%%%
The second method implements any LOCC protocol with finite
width per measurement and shared randomness.

\begin{theorem}
\label{thm:conv-slim}
Let ${\cal P}$ be a fine-grained LOCC protocol (see the end of 
Section \ref{primer}) implementing an instrument ${\cal I}$.
Then, ${\cal P}$ can be written as a convex combination 
${\cal P} = \sum_i \lambda_i {\cal P}^{(i)}$, where:
\begin{enumerate}
  \item each ${\cal P}^{(i)}$ is an LOCC protocol implementing some 
        instrument ${\cal L}^{(i)}$;
  \item each ${\cal P}^{(i)}$ has the same tree structure as ${\cal P}$;
  \item each edge CP map ${\cal L}^{(i)}_e$ is proportional to the 
        corresponding edge CP map ${\cal L}_e$ of ${\cal P}$;
  \item ${\cal L}_e = \sum_i \lambda_i {\cal L}^{(i)}_e$; 
  \item for each $i$ and each vertex $v$ associated with a local operation 
        on a $d_{\rm loc}$-dim system, at most $d_{\rm loc}^2$ outgoing edges 
        of $v$ have nonzero edge CP maps.
\end{enumerate}
\end{theorem}

To prove the above, we first describe and prove a corollary of Carath\'{e}odory's Theorem.

\begin{cor}[Improved Carath\'{e}odory's Theorem]
\label{ct2}
Let $H = \{v_i\} \subset \mathbb{R}^n$, $v \in \mathbb{R}^n$.
Consider the set of probability distributions $p$ on $H$ 
with barycentre $v$, i.e.
\[
  P(H;v) := \left\{ p \text{ p.d. s.t. } v = \sum_i p_i v_i \right\}.
\]
This set is closed and convex. Its extreme points have 
support cardinality at most $n+1$.
\end{cor}

This corollary allows us to write any original probability
distribution with barycentre $v$, which by definition is an element of
$P(H;v)$, as a convex combination of such extremal distributions, each of
which has support at most $n+1$.

\begin{proof-of}[of Corollary~\ref{ct2}]
The convexity and closedness are clear. We prove the statement
concerning the support cardinality of the extremal points of $P(H;v)$ 
via its contrapositive.  
Consider any given $q \in P(H;v)$ with support $S$ of size 
$|S| \geq n+2$.
This gives an expression of $v = \sum_{i\in S} q_i v_i$ in 
which all $q_i > 0$.  
But this just says that $v$ is in the convex hull of 
$\{v_i: i\in S\}$, so, we can use Lemma \ref{lem:ct} to express 
$v$ as a convex combination of at most $n+1$ elements of $S$, 
$v = \sum_{i\in S} r_i v_i$ with some $r_i=0$. 
Consider the relation 
\[
  v = \sum_{i\in S} q_i v_i = \sum_{i\in S} r_i v_i.
\]
Since for all $i\in S$, $q_i > 0$, there exists a $t>0$
such that for all $i \in S$, $q_i - t r_i \geq 0$.  
This gives $q = t r + (1-t)r'$ for some other probability 
distribution $r'$, which by linearity is also an element of $P(H;v)$.  
But $q \neq r$ since $q$ has support strictly 
larger than that of $r$, therefore, $q$ is not an extreme point of $P(H;v)$.
Taking the contrapositive, extreme points of $P(H;v)$ have support 
cardinality at most $n+1$.  
\end{proof-of}

A special case of the above corollary upper bounds 
the number of outcomes in extremal measurements (by choosing the 
$v_i$'s to be density matrices and $v$ to be the maximally mixed 
state).  

\begin{lemma}[Corollary 2.48 in \cite{Wat18}]
\label{POVM}
For any extremal measurement on a Hilbert space $\X$, there are at
most $\dim(\X)^2$ nonzero POVM elements.
\qed
\end{lemma}
 
This lemma was used in \cite{EDavies78}. Other sources for it
include \cite{PAR99} and \cite[Corollary~1]{DPP05}.

\begin{proof-of}[of Theorem~\ref{thm:conv-slim}]
As before, it suffices to consider the first measurement on $\X_1$  
made by the first party.  For any such measurement,
$$\mathcal{M}(\rho) = \sum_i A_i \, \rho A_i^\dagger \otimes \op{i}{i},$$
we can always consider the canonical form,
$$\mathcal{M}(\rho) = \sum_i \sqrt{A_i^{\dag}A_i} 
\rho \sqrt{A_i^\dagger A_i} \otimes \op{i}{i},$$ 
because there exists isometry $U_i$ such that $A_i=
U_i\sqrt{A_i^{\dag}A_i}$ and so the two measurements differ only by a
conditional isometry, which can be delayed to the next action round of
this party.  Thus, we only need to consider the POVM of each local
measurement.  

We can decompose this POVM as a convex combination of POVMs of
extremal measurements.  By Lemma \ref{POVM} above, 
each extremal measurement has no more than $d_1^2$ outcomes.

The same reasoning can be applied to subsequent measurements.  For
each vertex, the maps associated with the outgoing edges may take the
state to spaces of different dimensions.  To perform the induction
through the tree, we need to make the additional observation that, for
a fine-grained protocol, the range of each edge map has dimension no
bigger than the input dimension.
\end{proof-of}

\begin{theorem}
\label{thm:mainb}
Suppose an ensemble of multipartite quantum states
$S=\{p_1\rho_1,\cdots,p_n\rho_n\}\subset\L(\otimes_{j=1}^m\mathcal{\X}_j)$
with $p_k \geq 0$, $\sum_k p_k=1$ can be distinguished by some LOCC
protocol $\mathcal{P}$ with success probability $t$. Then there exists
an LOCC protocol $\mathcal{P}'$ achieving the same success probability
$t$ but in which each measurement requires at most $d_{\rm loc}^2$
outcomes, where $d_{\rm loc}$ is the dimension of the local system
measured.  If $\mathcal{P}$ has $\ell<\infty$ many rounds of 
communication, the total width can be bounded by $d^{2\ell}$, where 
$d$ is the maximum local dimension.
\end{theorem}

\begin{proof}
Because the probability of success is linear in the decomposition in
Theorem \ref{thm:conv-slim}, the best slim protocol has probability of
success at least $t$ (and at most $t$ by the optimality of the
original protocol).  So, we can replace the original protocol by this
slim protocol, in which each measurement on a system with local
dimension $d_{\rm loc}$ has no more than $d_{\rm loc}^2$ outcomes.
\end{proof}

From the above proof, it is evident that 
Theorem \ref{thm:mainb} applies to the maximization of any
function that is linear, or more generally convex, in the LOCC instrument.

To implement ${\cal P}$ via slim protocols as given by the 
decomposition in Theorem \ref{thm:conv-slim},
the parties need to share randomness.  The next
theorem bounds the required amount of shared randomness when 
the implemented instrument has finite input and output dimensions
and finitely many classical outcomes.

\begin{theorem}
\label{thm:shared-randomness}
Let ${\cal P}$ be a fine-grained LOCC protocol (see the end of 
Section \ref{primer}) implementing an instrument ${\cal I}$ with $n$
coarse-grained outcomes.  
Let $D_0$ be the
total input dimension, and $D_i$ be the total output dimension of 
the CP map conditioned
on the outcome
being $i$ for $i = 1,\cdots,n$.
Then, ${\cal I} = \sum_{i=1}^R \mu_i {\cal I}^{(i)}$, with each 
${\cal I}^{(i)}$ an instrument implemented by a slim protocol 
${\cal P}^{(i)}$ satisfying all the conditions in Theorem \ref{thm:conv-slim},
and $R \leq \sum_{i=1}^n D_0^2 D_i^2 - D_0^2 + 1$.
\end{theorem}

\begin{proof}
From Theorem \ref{thm:conv-slim}, 
${\cal I} = \sum_i \lambda_i {\cal I}^{(i)}$ where the $\lambda_i$'s 
form a probability distribution. The affine
space of instruments with $n$ coarse-grained outcomes and with the
stated input and output dimensions has dimension
$\sum_{i=1}^n D_0^2 D_i^2 - D_0^2$ (since the CP map corresponding 
to the $i$-th outcome is represented by a hermitian Choi matrix 
specified by $D_0^2 D_i^2$ real parameters, and the trace-preserving 
constraint removes $D_0^2$ real degrees of freedom.  
Applying Carath\'{e}odory's Theorem (Lemma \ref{lem:ct}), 
we can rewrite ${\cal I} = \sum_i \mu_i {\cal I}^{(i)}$
where at most $R=\sum_{i=1}^n D_0^2 D_i^2 - D_0^2 +1$ of
the $\mu_i$'s are nonzero.
So, $\log_2 R$ shared bits of randomness are sufficient.
\end{proof}

We note that for an LOCC protocol ${\cal P}$ represented by an
infinite tree, Theorem \ref{thm:shared-randomness} provides an exact
implementation of the corresponding LOCC instrument ${\cal I}$ as a
finite mixture of slim LOCC protocols, each of which can be
represented by a potentially infinite tree and each defines a bona
fide instrument, with probability $1$.  The $\ell$-round prefix of
this compressed infinite protocol converges to ${\cal P}$.

%%%%%%%%%%%%%%%%%%%%%%%%%%%%%%%%%%%%%%%%%%%%%%%%%%%%%%%%%%%%%%%%%%%%%%
\section{Acknowledgements}

We thank Scott Cohen, Anurag Anshu, Eric Chitambar, Laura Man\v{c}inska,
Dave Touchette, and John Watrous for helpful discussions.  DL was
supported by NSERC, CIFAR; AW was supported by the Spanish MINECO
(project FIS2016-86681-P) with the support of FEDER funds, and the
Generalitat de Catalunya (project 2017-SGR-1127); NY was supported by
DE180100156.

\newcommand{\etalchar}[1]{$^{#1}$}

\end{document}